\documentclass[11pt,a4paper,reqno,intlimits,sumlimits]{amsart}

\usepackage{amsmath}
\usepackage{amssymb}
\usepackage{chicago}
\usepackage[mathscr]{euscript}
\usepackage{enumerate}
\usepackage{xspace}
\usepackage{color}
\usepackage{epsfig,rotating}

\DeclareMathAlphabet{\mathpzc}{OT1}{pzc}{m}{it}

\usepackage[bookmarksopen,pdfstartview=FitH]{hyperref}
\hypersetup{colorlinks,%
           citecolor=black,%
           filecolor=black,%
           linkcolor=black,%
           urlcolor=black}

\begin{document}

\theoremstyle{plain}
\newtheorem{theorem}{Theorem}[section]
\newtheorem{lemma}[theorem]{Lemma}
\newtheorem{proposition}[theorem]{Proposition}
\newtheorem{corollary}[theorem]{Corollary}

\theoremstyle{definition}
\newtheorem{remark}[theorem]{Remark}
\newtheorem{example}[theorem]{Example}
\newtheorem{assumption}[theorem]{Assumption}

\newcommand{\Law}{\ensuremath{\mathop{\mathrm{Law}}}}
\newcommand{\loc}{{\mathrm{loc}}}
\newcommand{\Log}{\ensuremath{\mathop{\mathcal{L}\mathrm{og}}}}

\let\SETMINUS\setminus
\renewcommand{\setminus}{\backslash}

\def\stackrelboth#1#2#3{\mathrel{\mathop{#2}\limits^{#1}_{#3}}}

\renewcommand{\theequation}{\thesection.\arabic{equation}}
\numberwithin{equation}{section}

\newcommand{\prozess}[1][L]{{\ensuremath{#1=(#1_t)_{0\le t\le T}}}\xspace}
\newcommand{\prazess}[1][L]{{\ensuremath{#1=(#1_t)_{0\le t\le T^*}}}\xspace}
\newcommand{\scal}[2]{\ensuremath{\langle #1, #2 \rangle}}

\def\F{\ensuremath{\mathcal{F}}}
\def\R{\ensuremath{\mathbb{R}}}
\def\C{\ensuremath{\mathbb{C}}}
\def\bF{\mathbf{F}}
\def\V{\mathbb{V}}

\def\Rmz{\R\setminus\{0\}}
\def\Rdmz{\R^d\setminus\{0\}}
\def\Rnmz{\R^n\setminus\{0\}}
\def\Rp{\mathbb{R}_+}

\def\la{\ensuremath{L^1(\R)}}
\def\lad{\ensuremath{L^1(\R^d)}}
\def\lat{\ensuremath{L^2(\R)}}
\def\lap{\ensuremath{L^\infty(\R)}}
\def\labc{\ensuremath{L^1_{\text{bc}}(\R)}}

\def\lev{L\'{e}vy\xspace}
\def\lk{L\'{e}vy--Khintchine\xspace}
\def\smmg{semimartingale\xspace}
\def\mg{martingale\xspace}
\def\tih{time-inhomogeneous\xspace}
\def\chartri{\ensuremath{(b,\sigma,\nu)}}
\def\num{num\'{e}raire\xspace}

\def\eqlaw{\ensuremath{\stackrel{\mathrrefersm{d}}{=}}}

\def\dsdx{\ensuremath{(\ud s, \ud x)}}
\def\dtdx{\ensuremath{(\ud t, \ud x)}}

\def\intrr{\ensuremath{\int_{\R}}}

\def\ud{\ensuremath{\mathrm{d}}}
\def\e{\mathrm{e}}
\def\dt{\ud t}
\def\ds{\ud s}
\def\dx{\ud x}
\def\dy{\ud y}
\def\dz{\ud z}
\def\du{\ud u}
\def\icc{\mathpzc{i}}
\def\ecc{\mathbf{e}_\mathpzc{i}}

\def\EM{\ensuremath{(\mathbb{EM})}\xspace}
\def\ES{\ensuremath{(\mathbb{ES})}\xspace}
\def\AC{\ensuremath{(\mathbb{AC})}\xspace}

\def\ott{{0\leq t\leq T}}

\def\bg{\ensuremath{\bar g}}
\def\logs{\mathpzc s}

\title[Analysis of valuation formulas]
      {Analysis of fourier transform valuation formulas and applications}

\author[E. Eberlein]{Ernst Eberlein}
\author[K. Glau]{Kathrin Glau}
\author[A. Papapantoleon]{Antonis Papapantoleon}

\address{Department of Mathematical Stochastics, University of Freiburg,
        Eckerstr. 1, 79104 Freiburg, Germany}
\email{eberlein@stochastik.uni-freiburg.de}

\address{Department of Mathematical Stochastics, University of Freiburg,
        Eckerstr. 1, 79104 Freiburg, Germany}
\email{glau@stochastik.uni-freiburg.de}

\address{Institute of Mathematics, TU Berlin, Stra\ss e des 17. Juni 136,
         10623 Berlin, Germany \& Quantitative Products
         Laboratory, Deutsche Bank AG, Alexanderstr. 5, 10178 Berlin,
         Germany}
\email{papapan@math.tu-berlin.de}



\thanks{We would like to thank Friedrich Hubalek and Gabriel Maresch for
        valuable discussions. K.~G. would like to thank the DFG for financial
        support through project EB66/11-1, and the Austrian Science Fund (FWF)
        for an invitation under grant P18022. A.~P. gratefully acknowledges the
        financial support from the Austrian Science Fund (FWF grant Y328,
        START Prize). \\\indent
        We would like to thank the two anonymous referees for their careful
        reading of the manuscript and their valuable suggestions that have
        improved the paper.}

\date{}
\maketitle

\frenchspacing
\pagestyle{myheadings}

\begin{abstract}
The aim of this article is to provide a systematic analysis of the conditions
such that Fourier transform valuation formulas are valid in a general framework;
i.e. when the option has an arbitrary payoff function and depends on the path
of the asset price process. An interplay between the conditions on the payoff
function and the process arises naturally. We also extend these results to the
multi-dimensional case, and discuss the calculation of Greeks by Fourier
transform methods. As an application, we price options on the minimum of two
assets in \lev and stochastic volatility models.
\end{abstract}

\section{Introduction}

Since the seminal work of \citeN{CarrMadan99} and \citeN{Raible00} on the
valuation of options with Fourier transform methods, there have been several
articles dealing with extensions and analysis of these
valuation formulas. This literature focuses on the extension of the method to
other situations, e.g. the pricing of exotic or multi-asset derivatives, or
on the analysis of the discretization error of the fast Fourier transform.

The article of \citeN{BorovkovNovikov02} deals with the application of Fourier
transform valuation formulas for the pricing of some exotic options, while
\shortciteN{HubalekKallsenKrawczyk06} use similar techniques for hedging
purposes. Lee \citeyear{Lee04} provides an analysis of the discretization error in the
fast Fourier transform, while \citeN{Lord08} extends the method to the pricing
of options with early exercise features. Recently, \citeN{HubalekKallsen03},
Biagini et al.
\citeyear{BiaginiBregmanMeyerBrandis08} and \citeN{HurdZhou09} extend the
method to accommodate options on several assets, considering basket options,
spread options and catastrophe insurance derivatives.
\shortciteN{DufresneGarridoMorales09} also consider the valuation of payoffs
arising in insurance mathematics by Fourier methods. In addition, the books
of \citeN{ContTankov03}, Boyarchenko and Levendorski\v{\i}
\citeyear{BoyarchenkoLevendorskii02book} and
\citeN{Schoutens03} are also discussing Fourier transform methods for option
pricing. Let us point out that all these results are intimately related to
Parseval's formula, cf. \citeN[VI.2.2]{Katznelson04}.

The aim of our article is to provide a systematic analysis of the conditions
required for the \textit{existence} of Fourier transform valuation formulas in
a general framework: i.e. when the underlying variable can depend on the path
of the price process and the payoff function can be discontinuous. Such an
analysis seems to be missing in the literature.

In their work, \citeN{CarrMadan99}, \citeN{Raible00} and most others are
usually imposing a \emph{continuity} assumption, either on the payoff function
or on the random variable (i.e. existence of a Lebesgue density). However, when
considering e.g. a one-touch option on a \lev-driven asset, both  assumptions
fail: the payoff function is clearly discontinuous, while a priori not much is
known about the existence of a density for the distribution of the supremum of
a \lev process. Analogous situations can also arise in higher dimensions.

The key idea in Fourier transform methods for option pricing lies in the
separation of the \textit{underlying process} and the \textit{payoff function}.
We derive conditions on the moment generating function of the underlying random
variable and the Fourier transform of the payoff function such that
Fourier based valuation formulas hold true in one and several dimensions. An
interesting interplay between the continuity conditions imposed on the payoff
function and the random variable arises naturally. We also derive a result
that allows to easily verify the conditions on the payoff function (cf. Lemma
\ref{sobo}).

The results of our analysis can be briefly summarized as follows: for general
continuous payoff functions or for variables, whose distribution has a Lebesgue
density, the valuation formulas using Fourier transforms are valid as Lebesgue
integrals, in one and several dimensions. When the payoff function is discontinuous
and the random variable might not possess a Lebesgue density then, in dimension
one, we get pointwise convergence of the valuation formulas under additional
assumptions, that are typically satisfied. In several dimensions pointwise
convergence fails, but we can deduce the valuation function as an $L^2$-limit.

In addition, the structure of the valuation formulas allows us to derive easily
formulas for the sensitivities of the option price with respect to the various
parameters; otherwise, Malliavin calculus techniques or cubature formulas have
to be employed, cf. e.g. \shortciteN{FournieLasryLebuchouxLionsTouzi99},
\citeN{Teichmann06} and \citeN{KohatsuHigaYasuda08}. We  discuss results
regarding the sensitivities with respect to the initial value, i.e. the delta
and the gamma. It turns out that the trade-off between continuity conditions on
the payoff function and the random variable established for the valuation
formulas, becomes now a trade-off between integrability and smoothness
conditions for the calculation of the sensitivities.

The valuation formulas allow to compute prices of European options very fast,
hence they allow the efficient calibration of the model to market data for a
large variety of driving processes, such as \lev processes and affine stochastic
volatility models. Indeed, for \lev and affine processes the moment generating
function is usually known explicitly, hence these models are tailor-made for
Fourier transform pricing formulas.

We also mention here that the Fourier transform based approach can be applied
for the efficient computation of prices in other frameworks as well.
An important area is the
valuation of interest rate derivatives in \lev driven models. \lev term
structure models were developed in a series of papers in the last ten years;
this development is surveyed in \citeN{EberleinKluge06}.
For the Fourier based formulas we mention the two papers by
Eberlein and Kluge \citeyear{EberleinKluge04,EberleinKluge05}, where caps,
floors, and swaptions as well as interest rate digital and range digital options
are discussed; furthermore Eberlein and Koval \citeyear{EberleinKoval06},
where cross currency derivatives are considered and 
Eberlein, Kluge, and Sch\"onbucher \citeyear{EberleinKlugeSchoenbucher06}, where
pricing formulas for credit default swaptions are derived. Moreover, in the
framework of the `affine LIBOR' model (cf.
\shortciteNP{KellerResselPapapantoleonTeichmann09}) caps and swaptions can be
easily priced by Fourier based methods.

This paper is organized as follows: in Section \ref{gval} we present valuation
formulas in the single asset case, and in Section \ref{Rdval} we deal with the
valuation of options on several assets. In Section \ref{greeks} we discuss
sensitivities. In Section \ref{ch3:payoffs} we review examples of commonly
used payoff functions, in dimension one and in multiple dimensions. In Section
\ref{LA} we review L\'evy and affine processes. Finally, in Section
\ref{ch3:sc5} we provide numerical examples for the valuation of options on
several assets in \lev and affine stochastic volatility models.

\section{Option valuation: single asset}
\label{gval}

\subsection*{1.}
Let $\mathscr B=(\Omega,\F,\bF,P)$ be a stochastic basis in the sense of Jacod and 
Shiryaev \citeyear[I.1.3]{JacodShiryaev03}, where $\F=\F_T$ and $\bF
=(\F_t)_{0\le t\le T}$. We model the price process of a financial
asset, e.g. a stock or an FX rate, as an \textit{exponential
semimartingale} $S=(S_t)_{0\le t\le T}$, i.e. a stochastic process
with representation
\begin{equation}\label{ch3:eq1}
 S_t = S_0 \e^{H_t},\qquad 0\le t\le T
\end{equation}
(shortly: $S=S_0\e^H$), where $H=(H_t)_{0\le t\le T}$ is a
semimartingale with $H_0=0$.

Every semimartingale $H=(H_t)_{0\le t\le T}$
admits a \emph{canonical representation}
\begin{equation}\label{ch3:eq2}
 H = B + H^c+h(x)*(\mu-\nu) + (x-h(x))*\mu,
\end{equation}
where $h=h(x)$ is a \emph{truncation function}, $B=(B_t)_{0\le t\le T}$
is a predictable process of bounded variation, $H^c=(H_t^c)_{0\le t\le T}$
is the continuous martingale part of $H$ with predictable quadratic
characteristic $\langle H^c\rangle=C$, and $\nu$ is the predictable
compensator of the random measure of jumps $\mu$ of $H$. Here $W*\mu$
denotes the integral process of $W$ with respect to $\mu$, and $W*(\mu-\nu)$
denotes the stochastic integral of $W$ with respect to the compensated random
measure $\mu-\nu$; cf. \citeN[Chapter II]{JacodShiryaev03}.

Let $\mathcal M(P)$, resp. $\mathcal M_\loc(P)$, denote the class of
all martingales, resp. local martingales, on the given stochastic
basis $\mathscr B$.

Subject to the assumption that the process $1_{\{x>1\}}\e^x*\nu$ has
bounded variation, we can deduce the martingale condition
\begin{equation}\label{ch3:eq3}
 S = S_0\e^H\in\mathcal M_\loc(P)
  \Leftrightarrow
 B + \frac{C}{2} + (\e^x-1-h(x))*\nu = 0;
\end{equation}
cf. \shortciteN{EberleinPapapantoleonShiryaev06} for details. The martingale
condition can also be expressed in terms of the cumulant process $K$
associated to $(B,C,\nu)$, i.e. $K(1)=0$; for the cumulant process see
\citeN{JacodShiryaev03}.

Throughout this work, we assume that $P$ is an (equivalent) martingale
measure for the asset $S$ and the martingale condition is in force; moreover,
for simplicity we assume that the interest rate and dividend yield are zero.
By no-arbitrage theory the price of an option on $S$ is calculated as its
discounted expected payoff.

\subsection*{2.}
Let $Y=(Y_t)_{0\le t\le T}$ be a stochastic process on the given basis. We
denote by $\overline{Y}=(\overline{Y}_t)_{0\le t\le T}$ and
$\underline{Y}=(\underline{Y}_t)_{0\le t\le T}$ the supremum and the infimum
processes of $Y$ respectively, i.e.
\begin{displaymath}
 \overline{Y}_t = \sup_{0\leq u\leq t}Y_u
  \quad \mathrm{ and } \quad
 \underline{Y}_t = \inf_{0\leq u\leq t}Y_u.
\end{displaymath}

Notice that since the exponential function is monotonically increasing,
the supremum processes of $S$ and $H$ are related via
\begin{align}\label{sup-SH}
\overline{S}_T & = \sup_{0\leq t\leq T}\left(S_0\e^{H_t}\right)
                 = S_0\e^{\sup_{0\leq t\leq T} H_t}
                 = S_0\e^{\overline{H}_T}.
\end{align}
Similarly, the infimum processes of $S$ and $H$ are related via
\begin{align*}
\underline{S}_T = S_0\e^{\underline{H}_T}.
\end{align*}

\subsection*{3.}
The aim of this work is to tackle the problem of efficient valuation for
plain vanilla options, such as \emph{European call} and \emph{put} options,
as well as for exotic path-dependent options, such as \emph{lookback}
and \emph{one-touch} options, in a unified framework. Therefore, we will
analyze and prove valuation formulas for options on an asset $S=S_0\e^{H}$
with a payoff at maturity $T$ that may depend on the whole path of $S$ up
to time $T$. These results, together with analyticity conditions on the
Wiener--Hopf factors, will be used in the companion paper
\cite{EberleinGlauPapapantoleon09} for the pricing of one-touch and lookback
options in \lev models.

The following example of a fixed strike lookback option will serve as
a guideline for our methodology; note that using \eqref{sup-SH} it
can be re-written as
\begin{align}\label{stdlookb}
 \big(\overline{S}_T - K\big)^+ = \big(S_0\e^{\overline{H}_T} - K\big)^+\,.
\end{align}

In order to incorporate both plain vanilla options and exotic options in a
single framework we separate the \emph{payoff function} from the
\emph{underlying process}, where:
\begin{enumerate}[(a)]
\item the \emph{underlying process} can be the log-asset price
      process or the supremum/infimum of the log-asset price process
      or an average of the log-asset price process. This process will always be
      denoted by $X$ (i.e. $X=H$ or $X=\overline H$ or $X=\underline H$, etc.);
\item the \emph{payoff function} is an arbitrary function
      $f:\R\rightarrow\Rp$, for example $f(x)=(\e^x-K)^+$ or
      $f(x)=1_{\{\e^x>B\}}$, for $K,B\in\Rp$.
\end{enumerate}

Clearly, we regard options as dependent on the \textit{underlying process} $X$,
i.e. on (some functional of) the logarithm of the asset price process $S$. The
main advantage is that the characteristic function of $X$ is easier to handle
than that of (some functional of) $S$; for example, for a \lev process $H=X$ it
is already known in advance.

Moreover, we consider exactly those options where we can incorporate the
path-dependence of the option payoff into the underlying process $X$. European
vanilla options are a trivial example, as there is no path-dependence; a
non-trivial, example are options on the supremum, see again
\eqref{sup-SH} and \eqref{stdlookb}. Other examples are the geometric Asian
option and forward-start options.

In addition, we will assume that the initial value of the underlying process $X$
is zero; this is the case in all natural examples in mathematical finance. The
initial value $S_0$ of the asset price process $S$ plays a particular role,
because it is convenient to consider the option price as a function of it, or
more specifically as a function of $\logs=-\log S_0$.

Hence, we express a general payoff as
\begin{align}
\Phi\big(S_0\e^{H_t},\,0\le t\le T\big) = f(X_T-\logs)\,,
\end{align}
where $f$ is a payoff function and $X$ is the underlying process, i.e.
an adapted process, possibly depending on the full history of $H$,
with
$$X_t:= \Psi(H_s,\,0\le s\le t)\quad \text{ for } \,t\in[0,T],$$
and $\Psi$ a measurable functional. Therefore, the time-$0$ price of the option
is provided by the (discounted) expected payoff, i.e.
\begin{align}\label{genericprice}
\mathbb{V}_f(X;\logs)
 = E\big[\Phi\big(S_t,\,0\le t\le T\big)\big]
 = E\big[ f(X_T-\logs)\big]\,.
\end{align}

Note that we consider `European style' options, in the sense that the holder
or writer do not have the right to exercise or terminate the option before
maturity.

\begin{remark}
In case the interest rate $r$ and the dividend yield $\delta$ are non-zero,
then the martingale condition \eqref{ch3:eq3} reads
\begin{equation*}
 (\delta - r)t + B_t + \frac{C_t}{2} + (\e^x-1-h(x))*\nu_t = 0
\end{equation*}
for all $t$, and the option price is given by
$\mathbb{V}_f(X;\logs)
 = \e^{-rT}E\big[ f(X_T-\logs)\big]\,.$
\end{remark}

\subsection*{4.}
The first result focuses on options with continuous payoff functions, such as
European plain vanilla options, but also lookback options.

Let $P_{X_T}$ denote the law, $M_{X_T}$ the moment generating function and
$\varphi_{X_T}$ the (extended) characteristic function of the random variable
$X_T$; that is
\begin{align*}
M_{X_T}(u)
 = E\big[\e^{uX_T}\big]
 = \varphi_{X_T}(-iu),
\end{align*}
for suitable $u\in\C$. For any payoff function $f$ let $g$ denote the \emph{dampened}
payoff function, defined via
\begin{align}\label{g-defin}
g(x)=\e^{-Rx}f(x)
\end{align}
for some $R\in\R$. Let $\widehat{g}$ denote the (extended) Fourier transform
of a function $g$, and $L^1_{\text{bc}}(\R)$ the space of bounded, continuous
functions in $L^1(\R)$.

In order to derive a valuation formula for an option with an
arbitrary \emph{continuous} payoff function $f$, we will impose the
following conditions.
\begin{description}
\item[(C1)] Assume that $g\in L^1_{\text{bc}}(\R)$.
\item[(C2)] Assume that $M_{X_T}(R)$ exists.
\item[(C3)] Assume that $\widehat{g}\in L^1(\R)$.
\end{description}

\begin{theorem}\label{valuation}
If the asset price process is modeled as an exponential
semimartingale process according to \eqref{ch3:eq1}--\eqref{ch3:eq3}
and conditions (C1)--(C3) are in force, then the time-0 price function
is given by
\begin{align}\label{value}
\mathbb V_f(X;\logs) =
 \frac{\e^{-R\logs}}{2\pi}
  \int_\R \e^{-iu\logs}\varphi_{X_T}(u-iR)\widehat{f}(iR-u)\ud u.
\end{align}
\end{theorem}
\begin{proof}
Using \eqref{genericprice} and \eqref{g-defin} we have
\begin{align}\label{ch3:eq4}
\V_f(X;\logs)
    = \int_\Omega f(X_T-\logs)\ud P
    = \e^{-R\logs}\int_\R \e^{Rx}g(x-\logs)P_{X_T}(\dx).
\end{align}
By assumption (C1), $g\in\la$ and the Fourier transform of $g$,
\begin{align*}
 \widehat{g}(u) = \int_\R \e^{iux} g(x)\dx,
\end{align*}
is well defined for every $u\in\R$ and is also continuous and
bounded. Additionally, using assumption (C3) we immediately
have that $\widehat{g}\in\labc$.  Therefore, using the Inversion
Theorem (cf. \citeNP[Theorem 3.4.4]{Deitmar04}), $\widehat{g}$
can be inverted and $g$ can be represented, for \emph{all} $x\in\R$, as
\begin{align}\label{ch3:eq7}
 g(x) = \frac1{2\pi} \int_\R \e^{-ixu}\widehat{g}(u)\ud u.
\end{align}

Now, returning to the valuation problem \eqref{ch3:eq4} we get that
\begin{align}\label{ch3:eq8}
\V_f(X;\logs)
   &= \e^{-R\logs}\int_\R \e^{Rx}
      \Bigg(\frac1{2\pi}\int_\R \e^{-i(x-\logs)u}\widehat{g}(u)\ud u\Bigg)P_{X_T}(\dx)\nonumber\\
   &= \frac{\e^{-R\logs}}{2\pi}\int_\R \e^{iu\logs}
      \Bigg( \int_\R \e^{i(-u-iR)x}P_{X_T}(\dx)\Bigg)\widehat{g}(u)\ud u\nonumber\\
   &= \frac{\e^{-R\logs}}{2\pi}\int_\R \e^{iu\logs}\varphi_{X_T}(-u-iR)\widehat{f}(u+iR)\ud u,
\end{align}
where for the second equality we have applied Fubini's theorem;
moreover, for the last equality we have
\begin{align*}
\widehat{g}(u) = \int_\R \e^{iux}\e^{-Rx}f(x)\dx
                 = \widehat{f}(u+iR).
\end{align*}

Finally, the application of Fubini's theorem is justified since
\begin{align*}
\int_\R\int_\R \e^{Rx}|\e^{-iu(x-\logs)}||\widehat{g}(u)|\ud u P_{X_T}(\dx)
  &\le \int_\R\e^{Rx}\bigg(\int_\R |\widehat{g}(u)|\ud u \bigg)P_{X_T}(\dx)\\
  &\le K M_{X_T}(R) < \infty,
\end{align*}
where we have used again that $\widehat{g}\in L^1(\R)$, and the finiteness
of $M_{X_T}(R)$ is given by Assumption (C2).
\end{proof}

\begin{remark}
We could also replace assumptions (C1) and (C3) with the following conditions
\begin{center}
(C1$'$): $g\in L^1(\R)\quad$ and $\quad$
(C3$'$): $\widehat{\e^{Rx}P_{X_T}}\in L^1(\R)$.
\end{center}
Condition (C3$'$) yields that $\e^{Rx}P_{X_T}$ possesses a continuous
bounded Lebesgue density, say $\rho$; cf. \citeN[Theorem 8.39]{Breiman68}.
Then, we can identify $\rho$, instead of $g$, with the inverse of its Fourier
transform and the proof goes through with the obvious modifications. This
statement is almost identical to Theorem 3.2 in \citeN{Raible00}.
\end{remark}

\begin{remark}[Numerical evaluation]
The option price represented as an integral of the form \eqref{value} can be
evaluated numerically very fast. The following simple observation can speed
up the computation of this expression even further: notice that for a fixed
maturity $T$, the characteristic function -- which is the computationally
expensive part -- should only be evaluated \emph{once} for \emph{all}
different strikes or initial values. The gain in computational time will be
significant when considering models where the characteristic function is not
known in closed form; e.g. in affine models where one might need to solve a
Riccati equation to obtain the characteristic function. This observation has
been termed `caching' by some authors (cf. \citeNP{Kilin07})
\end{remark}

\subsection*{5.}
Apart from (C3), the prerequisites of Theorem \ref{valuation} are quite easy
to check in specific cases. In general, it is also an interesting question
to know when the Fourier transform of an integrable function is integrable.
The problem is well understood for smooth ($C^2$ or $C^\infty$) functions, see
e.g. \citeN{Deitmar04}, but the functions we are dealing with are typically
not smooth. Hence, we will provide below an easy-to-check condition for a
non-smooth function to have an integrable Fourier transform.

Let us consider the Sobolev space $H^1(\R)$, with
\begin{align*}
 H^1(\R) = \Big\{g\in L^2(\R) \;\Big|\; \partial g \text{ exists and } \partial g\in L^2(\R)\Big\},
\end{align*}
where $\partial g$ denotes the \emph{weak} derivative of a function
$g$; see e.g. \citeANP{Sauvigny06} \citeyear{Sauvigny06}. Let $g\in H^1(\R)$, then from
Proposition 5.2.1 in \citeN{Zimmer90} we get that
\begin{align}\label{fou-der}
 \widehat{\partial g}(u) = -iu\widehat{g}(u)
\end{align}
and $\widehat{g},\widehat{\partial g}\in L^2(\R)$.

\begin{lemma}\label{sobo}
Let $g\in H^1(\R)$, then $\widehat{g}\in L^1(\R)$.
\end{lemma}
\begin{proof}
Using the above results, we have that
\begin{align}\label{fou-fdf}
 \infty
  > \int_\R \Big(\big|\widehat{g}(u)\big|^2 + \big|\widehat{\partial g}(u)\big|^2\Big)\ud u
  = \int_\R \big|\widehat{g}(u)\big|^2\big(1+|u|^2\big)\ud u.
\end{align}
Now, by the H\"older inequality, using $(1+|u|)^2\leq3(1+|u|^2)$
and \eqref{fou-fdf}, we get that
\begin{align*}
\int_\R \big|\widehat{g}(u)\big|\ud u
 &=   \int_\R \big|\widehat{g}(u)\big|\frac{1+|u|}{1+|u|}\ud u\\
 &\le \bigg(\int_\R \big|\widehat{g}(u)\big|^2(1+|u|)^2\ud u\bigg)^{\frac12}
      \bigg(\int_\R \frac{1}{(1+|u|)^2}\ud u\bigg)^{\frac12}
  < \infty
\end{align*}
and the result is proved.
\end{proof}

\begin{remark}
A similar statement can be proved for functions in the Sobolev-Slobodeckij
space $H^s(\R)$, for $s>\frac12$.
\end{remark}

\subsection*{6.}
Next, we deal with the valuation formula for options whose payoff function can
be \textit{discontinuous}, while at the same time the measure $P_{X_T}$ does
\textit{not} necessarily possess a Lebesgue density. Such a situation arises
typically when pricing one-touch options in purely discontinuous \lev models.
Hence, we need to impose different conditions, and we derive the valuation formula
as a pointwise limit by generalizing the proof of Theorem 3.2 in \citeN{Raible00}.
A similar result (Theorem 1 in \shortciteNP{DufresneGarridoMorales09}) has been
pointed out to us by one of the referees.

In this and the following sections we will make use of the following
notation; we define the function $\bg$ and the measure $\varrho$ as
follows
\begin{align*}
\bg(x):=g(-x)\quad \text{ and } \quad \varrho(\dx):=\e^{Rx}P_{X_T}(\dx).
\end{align*}
Moreover $\varrho(\R)=\int\varrho(\dx)$, while $\bg*\varrho$ denotes the
convolution of the function $\bg$ with the measure $\varrho$. In this case
we will use the following assumptions.

\begin{description}
\item[(D1)] Assume that $g\in L^1(\R)$.
\item[(D2)] Assume that $M_{X_T}(R)$ exists
            ($\Longleftrightarrow \varrho(\R)<\infty$).
\end{description}

\begin{theorem}\label{valuation-dc}
Let the asset price process be modeled as an exponential
semimartingale process according to \eqref{ch3:eq1}--\eqref{ch3:eq3}
and conditions (D1)--(D2) be in force. The time-0 price function is
given by
\begin{align}\label{value-dc}
\mathbb V_f(X;\logs) =
  \lim_{A\rightarrow\infty}\frac{\e^{-R\logs}}{2\pi}
  \int_{-A}^A \e^{-iu\logs}\varphi_{X_T}(u-iR)\widehat{f}(iR-u)\ud u,
\end{align}
at the point $\logs\in\R$, if $\mathbb V_f(X;\cdot)$ is of bounded
variation in a neighborhood of $\logs$, and $\mathbb V_f(X;\cdot)$
is continuous at $\logs$.
\end{theorem}

\begin{remark}
In Section \ref{ch3:payoffs} we will relate the conditions on the
valuation function $\mathbb V_f$ to properties of the measure $P_{X_T}$
for specific (dampened) payoff functions $g$. These properties are
easily checkable -- and typically satisfied -- in many models.
\end{remark}

\begin{proof}
Starting from \eqref{ch3:eq4}, we can represent the option price function
as a convolution of $\bg$ and $\varrho$ as follows
\begin{align}\label{value-pw}
\V_f(X;\logs) &= \e^{-R\logs}\int_\R \e^{Rx}g(x-\logs)P_{X_T}(\dx)
               = \e^{-R\logs}\int_\R \bg(\logs-x)\varrho(\dx) \nonumber\\
              &= \e^{-R\logs}\bg*\varrho(\logs).
\end{align}
Using that $g\in L^1(\R)$, hence also $\bg\in L^1(\R)$, and $\varrho(\R)<\infty$
we get that $\bg*\varrho\in L^1(\R)$, since
\begin{align}\label{L1-conv}
\Vert\bg*\varrho\Vert_{L^1(\R)}
 &\le \varrho(\R)\,\Vert\bg\Vert_{\la}<\infty;
\end{align}
compare with Young's inequality, cf. \citeN[IV.1.6]{Katznelson04}.
Therefore, the Fourier transform of the convolution is well defined
and we can deduce that, for all $u\in\R$,
\begin{align*}
 \widehat{\bg*\varrho}(u)=\widehat{\bg}(u)\cdot \widehat{\varrho}(u);
\end{align*}
compare with Theorem 2.1.1 in \citeN{Bochner55}.

By \eqref{L1-conv} we can apply the inversion theorem for the Fourier
transform, cf. Satz 4.2.1 in \citeN{Doetsch50}, and get
\begin{align}\label{ug-for}
\frac12 \big(\bg*\varrho(\logs+)+\bg*\varrho(\logs-)\big)
 &= \frac{1}{2\pi} \lim_{A\rightarrow\infty}
    \int_{-A}^A \e^{-iu\logs}\widehat{\varrho}(u)\widehat{\bg}(u)\du,
\end{align}
if there exists a neighborhood of $\logs$ where
$\logs\mapsto\bg*\varrho(\logs)$ is of bounded variation.

We proceed as follows: first we show that the function
$\logs\mapsto\bg*\varrho(\logs)$ has bounded variation; then we show
that this map is also continuous, which yields that the left hand side
of \eqref{ug-for} equals $\bg*\varrho(\logs)$.

For that purpose, we re-write \eqref{value-pw} as
\begin{align*}
\bg*\varrho(\logs) = \e^{R\logs}\,\mathbb V_f(X;\logs);
\end{align*}
then, $\bg*\varrho$ is of bounded variation on a compact interval
$[a,b]$ if and only if $\mathbb V_f(X;\cdot)\in BV([a,b])$;
this holds because the map $\logs\mapsto\e^{R\logs}$ is of bounded
variation on any bounded interval on $\R$, and the fact that the
space $BV([a,b])$ forms an algebra; cf. Satz 91.3 in \citeN{Heuser01}.
Moreover, $\logs$ is a continuity point of $\bg*\varrho$ if and only
if $\mathbb V_f(X;\cdot)$ is continuous at $\logs$.

In addition, we have that
\begin{align}\label{dc-val-3}
\widehat{\bg}(u)
 &= \int_\R \e^{-iux}\e^{-Rx}f(x)\dx
  = \widehat{f}(iR-u)
\end{align}
and
\begin{align}\label{dc-val-4}
\widehat{\varrho}(u)
  = \int_\R \e^{iux}\e^{Rx}P_{X_T}(\dx)
  = \varphi_{X_T}(u-iR).
\end{align}
Hence, \eqref{ug-for} together with \eqref{dc-val-3}, \eqref{dc-val-4}
and the considerations regarding the continuity and bounded variation
properties of the value function yield the required result.
\end{proof}

\section{Option valuation: multiple assets}
\label{Rdval}

\subsection*{1.}
We would like to establish valuation formulas for options that depend
on several assets or on multiple functionals of one asset.
Typical examples of options on several assets are basket options and
options on the minimum or maximum of several assets, with payoff
\begin{displaymath}
  (S_T^1\wedge\cdots\wedge S_T^d-K)^+,
\end{displaymath}
where $x\wedge y=\min\{x,y\}$. Typical examples of options on functionals
of a single asset are \emph{barrier} options, with payoff
\begin{displaymath}
  (S_T-K)^+1_{\{\overline S_T>B\}},
\end{displaymath}
and \emph{slide-in} or \emph{corridor} options, with payoff
\begin{displaymath}
  (S_T-K)^+ \sum_{i=1}^N1_{\{L<S_{T_i}<H\}},
\end{displaymath}
at maturity $T$, where $0=T_0<T_1<\dots<T_N=T$.

In the previous section we proved that the valuation formulas for a single
underlying is still valid -- at least as a pointwise limit, under reasonable
additional assumptions -- even if the underlying distribution does not possess
a Lebesgue density and the payoff is discontinuous.

In the present section we will generalize the valuation formulas to the case
of several underlyings. Once again, if either the joint distribution possesses
a Lebesgue density or the payoff function is continuous, the formula is valid
as a Lebesgue integral. In case both assumptions fail, we will encounter situations
that are apparently of harmless nature, but where the pointwise convergence
will fail. In this case we will establish the valuation formulas as an $L^2$-limit;
however, with respect to numerical evaluation, a stronger notion of convergence
would be preferable.

Analogously to the single asset case we assume that the asset prices evolve
as exponential semimartingales. Let the driving process be an $\R^d$-valued
semimartingale $H=(H^1,\dots,H^d)^\top$ and $S=(S^1,\dots,S^d)^\top$ be the
vector of asset price processes; then each component $S^\icc$ of $S$ is
modeled as an exponential semimartingale, i.e.
\begin{align}\label{asset-Rd}
S^\icc_{t}=S^\icc_0\exp H^\icc_{t}, \qquad \ott, \;1\le \icc\le d,
\end{align}
where $H^\icc$ is an $\R$-valued
semimartingale with canonical representation
\begin{equation}\label{Hi-canon}
  H^\icc = H^\icc_0 + B^\icc + H^{\icc,c} + h^\icc(x)*(\mu-\nu) + (x^\icc-h^\icc(x))*\mu,
\end{equation}
with $h^\icc(x)=\e_\icc^\top h(x)$. The martingale condition can be given
as in eq. (3.3) in \citeN{EberleinPapapantoleonShiryaev08}.

\subsection*{2.}
In the sequel, we will price options with payoff $f(X_T-\logs)$
at maturity $T$, where $X_T$ is an $\F_T$-measurable $\R^d$-valued
random variable, possibly dependent of the history of the $d$
driving processes, i.e.
\begin{align*}
 X_T = \Psi\big(H_t,\,0\le t\le T\big),
\end{align*}
where $\Psi$ is an $\R^d$-valued measurable functional. Further $f$ is a
measurable function $f:\R^d\rightarrow\R_+$, and
$\logs=(\logs^1,\dots,\logs^d)\in\R^d$ with $\logs^\icc=-\log S_0^\icc$.

Analogously to the single asset case, we use the dampened payoff function
\begin{align*}
 g(x):= \e^{-\langle R,x\rangle}f(x) \quad\text{ for } x\in\R^d,
\end{align*}
and denote by $\varrho$ the measure defined by
\begin{align*}
 \varrho(\dx):= \e^{\langle R,x\rangle} P_{X_T}(\dx),
\end{align*}
where $R\in\R^d$ serves as a dampening coefficient. Here $\langle\cdot,\cdot\rangle$
denotes the Euclidian scalar product in $\R^d$. The scalar product is extended to
$\C^d$ as follows: for $u,v\in\C^d$, set $\langle u,v\rangle = \sum_iu_iv_i$, i.e.
we do not use the Hermitian inner product. Moreover, $M_{X_T}$ and
$\varphi_{X_T}$ denote the moment generating, resp. characteristic, function
of the random vector $X_T$.

To establish our results we will make use of the following assumptions.

\begin{description}
\item[(A1)] Assume that $g\in L^1(\R^d)$.
\item[(A2)] Assume that $M_{X_T}(R)$ exists.
\item[(A3)] Assume that $\widehat\varrho\in L^1(\R^d)$.
\end{description}

\begin{remark}
We can also replace Assumptions (A1) and (A3) with the following
assumption
\begin{description}
\item[(A1$'$)] Assume that $g\in L_{\text{bc}}^1(\R^d)$ and $\widehat{g}\in L^1(\R^d)$;
\end{description}
this shows again the interplay between the continuity properties of the payoff
function and the underlying distribution.
\end{remark}

\begin{theorem}\label{valuation-Rd}
If the asset price processes are modeled as exponential semimartingale
processes according to \eqref{asset-Rd}--\eqref{Hi-canon} and conditions
(A1)--(A3) are in force, then the time-$0$ price function is given by
\begin{align}\label{value-Rd}
\mathbb V_{f}(X;\logs)
 &= \frac{\e^{-\langle R,\logs\rangle}}{(2\pi)^d}
    \int_{\R^d} \e^{-i\langle u,\logs\rangle} M_{X_T}(R+iu) \widehat{f}(iR-u) \ud u.
\end{align}
\end{theorem}
\begin{proof}
Similarly to the one-dimensional case we have that
\begin{align}\label{value-map}
\V_f(X;\logs) = \e^{-\langle R,\logs\rangle} \bar{g}*\varrho(\logs).
\end{align}
Since $g\in L^1(\R^d)$ and $\varrho(\R^d)<\infty$, we get that
$\bar{g}*\varrho\in L^1(\R^d)$; therefore
$\widehat{\bg*\varrho}(u)=\widehat{\bg}(u)\cdot\widehat{\varrho}(u)$ for all
$u\in\R^d$. By assumption we know that $\widehat{\varrho}\in L^1(\R^d)$; moreover
$\widehat{\bg}\in L^\infty(\R^d)$ since
$|\widehat{\bg}|\le\Vert g\Vert_{L^1(\R^d)}<\infty$.
These considerations yield that $\widehat{\bg*\varrho}\in L^1(\R^d)$, again by
using Young's inequality.

Hence, applying the formula for the Fourier inversion, cf. Corollary
1.21 in \citeN{SteinWeiss71}, we conclude that
\begin{align*}
\V_f(X;\logs)
 &= \frac{\e^{-\langle R,\logs\rangle}}{(2\pi)^d}
    \int_{\R^d}\e^{-i\langle u,\logs\rangle}\widehat{\bg}(u)\widehat{\varrho}(u)\du \nonumber\\
 &= \frac{\e^{-\langle R,\logs\rangle}}{(2\pi)^d}
    \int_{\R^d} \e^{-i\langle u,\logs\rangle} M_{X_T}(R+iu) \widehat{f}(iR-u) \ud u,
\end{align*}
for a.e. $\logs\in\R^d$.

Moreover, if $\logs\mapsto\V_f(X;\logs)$ is continuous, then the
equality holds pointwise for \textit{all} $\logs\in\R^d$. The mapping
\eqref{value-map} is continuous if the mapping $\logs\mapsto\bg*\varrho(\logs)$
is continuous. Using Assumption (A3) we have that $\varrho$ possesses
a bounded continuous Lebesgue density $\rho\in L^1(\R^d)$; cf. Proposition
2.5 (xii) in Sato \citeyear{Sato99}. Then $\bg*\varrho=\bg*\rho$ and
\begin{align}\label{Rd-cont}
\lim_{|x|\rightarrow0} \bg*\varrho(\logs+x)
 &= \lim_{|x|\rightarrow0}\int \bg(\logs+x-z)\rho(z)\dz\nonumber\\
 &= \int\lim_{|x|\rightarrow0} \bg(\logs+y)\rho(x-y)\dy
  = \bg*\varrho(\logs)
\end{align}
yielding the continuity of the map. Note that we have used the continuity of
$\rho$; additionally, we can interchange integration and limit using the dominated
convergence theorem, with majorant $\bg(\cdot)\max_x\rho(x)$.
\end{proof}

\begin{remark}
The proof using Assumption (A1$'$) follows analogously, with the obvious
modifications for \eqref{Rd-cont}.
\end{remark}

\subsection*{3.}
Next, we consider the valuation of options on several assets when the
payoff function is \emph{discontinuous} and the driving process does
\textit{not} necessarily possess a Lebesgue density.

The main difference to the analogous situation in dimension one is that
the pointwise convergence of capped Fourier integrals -- as is the case
in Satz 4.2.1 in \citeN{Doetsch50} -- cannot be generalized to the
multidimensional case. M. Pinsky gives the following astonishing example
to illustrate this fact, see section 4.1 in \citeN{Pinsky93}; let $f$ be
the indicator function of the unit ball in $\R^3$, then
\begin{align}\label{pin-couex}
\frac{1}{(2\pi)^{3}}\int_{|x|\le A}\e^{-i\langle u,x\rangle}\widehat{f}(x)\dx\Big|_{u=0}
 = 1 - \frac{2}{\pi} \sin(A) + o(1),
\end{align}
for $A\uparrow \infty$. Extrapolating the convergence results from the
one-dimensional case to $\R^3$, we would expect pointwise convergence
of the spherical sum to the indicator function, at least in the interior
of the ball; on the contrary, the right hand side of \eqref{pin-couex}
is even \textit{divergent}.

As a consequence, we only derive an $L^2$-limit for the valuation function.

The setting is similar to the previous sections, and we need to impose the
following conditions.
\begin{description}
\item[(G1)] Assume that $g\in L^1(\R^d)\cap L^2(\R^d)$.
\item[(G2)] Assume that $M_{X_T}(R)$ exists.
\end{description}

\begin{theorem}\label{L2-valuation}
If the asset price process is modeled as an exponential semimartingale
process according to \eqref{asset-Rd}--\eqref{Hi-canon} and conditions
(G1)--(G2) are in force, then the time-$0$ price function satisfies
\begin{align}\label{L2-value}
\mathbb V_f(X;\cdot) =
 \frac{\e^{-\langle R,\cdot\rangle}}{(2\pi)^d}
 \mathop{L^2\hbox{-}\lim}_{A\rightarrow\infty}
 \int_{[-A,A]^d}\e^{-i\langle u,\cdot\rangle}\varphi_{X_T}(u-iR)\widehat{f}(iR-u)\ud u.
\end{align}
\end{theorem}
\begin{proof}
Similarly to the previous section, we have that
\begin{align}\label{ch3:sec5.1}
\V_f(X;\logs) = \e^{-\langle R,\logs\rangle}\bg*\varrho(\logs),
\end{align}
and, for all $u\in\R^d$
\begin{align}\label{L2-fouval}
 \widehat{\bg*\varrho}(u)=\widehat{\bg}(u)\cdot\widehat{\varrho}(u).
\end{align}
Now, since $\bg\in L^1(\R^d)\cap L^2(\R^d)$, we get that
$\widehat{\bg}\in L^2(\R^d)$ and
$\Vert\bg\Vert_{L^2(\R^d)}=\Vert\widehat{\bg}\Vert_{L^2(\R^d)}$; the
proofs are analogous to Theorem 9.13 in \citeN{Rudin87}.
Moreover, we have that $\bg*\varrho\in L^2(\R^d)$, because
\begin{align*}
\Vert\bg*\varrho\Vert^2_{L^2(\R^d)}
 &\le \varrho(\R^d)^2\,\Vert\bg\Vert^2_{L^2(\R^d)}<\infty.
\end{align*}
Therefore, since also $\bg*\varrho\in L^1(\R^d)$ we get that
$\bg*\varrho\in L^1(\R^d)\cap L^2(\R^d)$
and, analogously again to Theorem 9.13 in \citeN{Rudin87}, we get
that $\widehat{\bg*\varrho}\in L^2(\R^d)$
and $\Vert\bg*\varrho\Vert_{L^2(\R^d)}=\Vert\widehat{\bg*\varrho}\Vert_{L^2(\R^d)}$.

Therefore, the Fourier transform in \eqref{L2-fouval} can be  inverted
and the inversion is given as an $L^2$-limit; more precisely, we have
\begin{align}\label{dc-val-1}
\Vert\bg*\varrho-\psi_A\Vert_{L^2(\R^d)}\rightarrow0
\qquad (A\rightarrow\infty)
\end{align}
where
\begin{align}\label{dc-val-2}
\psi_A(\logs)
 &= \frac{1}{(2\pi)^d}\int_{[-A,A]^d}
    \!\e^{-i\langle u,\logs\rangle}\widehat{\bg*\varrho}(u)\ud u\nonumber\\
 &= \frac{1}{(2\pi)^d}\int_{[-A,A]^d}\!\e^{-i\langle u,\logs\rangle}
    \widehat{f}(iR-u)\varphi_{X_T}(u-iR)\ud u.
\end{align}
Finally, \eqref{ch3:sec5.1} and \eqref{dc-val-1}--\eqref{dc-val-2}
yield the option price function.
\end{proof}

\begin{remark}
The problem becomes significantly simpler when dealing with the product
$f_1(X_T)f_2(Y_T)$ of a continuous payoff function $f_1$ for the variable
$X$ and a discontinuous payoff function $f_2$ for the other variable $Y$,
even in the absence of Lebesgue densities. A typical example of this situation
is the barrier option payoff, where $f_1(x)=(\e^{x}-K)^+$ and
$f_2(y)=1_{\{\e^{y}>B\}}$. Then, one can make a measure change using the
(normalized) continuous payoff as the Radon--Nikodym derivative, apply Theorem
\ref{valuation} and then Theorem \ref{valuation-dc}; this leads to pointwise
convergence of the valuation function. The measure change argument is outlined
in \citeN{BorovkovNovikov02} and \citeN[Theorem 3.5]{Papapantoleon06}.
\end{remark}

\section{Sensitivities -- Greeks}
\label{greeks}

The structure of the asset price model as an exponential semimartingale,
and the resulting structure of the option price function, allows us to
easily derive general formulas for the sensitivities of the option price
with respect to model parameters. In this section we will focus on the
sensitivities with respect to the initial value, i.e. delta and gamma,
while sensitivities with respect to other parameters can be derived
analogously.

Let us rewrite the option price function as a function of the initial
value, using that $S_0=\e^{-\logs}$, as follows:
\begin{align}\label{price}
\mathbb{V}_f(X;S_0)
 &= \frac{1}{2\pi}\int_{\R} S_0^{R-iu} M_{X_T}(R-iu) \widehat{f}(u+iR) \ud u.
\end{align}
The delta of an option is the partial derivative of
the price with respect to the initial value. For a generic option with payoff
$f$, we have that
\begin{align}\label{delta}
\Delta_f(X;S_0)
 &= \frac{\partial\mathbb{V}_f(X;S_0)}{\partial S_0}\nonumber\\
 &= \frac{1}{2\pi}\int_{\R} \frac{\partial}{\partial S_0}S_0^{R-iu}
     M_{X_T}(R-iu) \widehat{f}(u+iR) \ud u \nonumber\\
 &= \frac{1}{2\pi}\int_{\R} S_0^{R-1-iu} M_{X_T}(R-iu)
     \frac{\widehat{f}(u+iR)}{(R-iu)^{-1}}\ud u.
\end{align}
The gamma of an option is the partial derivative of the delta with respect
to the initial value. For a generic option with payoff $f$, we  get
\begin{align}\label{gamma}
\Gamma_f(X;S_0)
 &= \frac{\partial\Delta_f(X;S_0)}{\partial S_0}
  = \frac{\partial^2\mathbb{V}_f(X;S_0)}{\partial^2 S_0}\nonumber\\
 &= \frac{1}{2\pi}\int_{\R} S_0^{R-2-iu}
     \frac{M_{X_T}(R-iu)\widehat{f}(u+iR)}{(R-1-iu)^{-1}(R-iu)^{-1}}\ud u.
\end{align}

In the above equations we have taken \textit{for granted} that we can exchange
integration and differentiation; however, this is the crucial step and we will
provide sufficient conditions when we are allowed to do that. Using Satz
IV.5.7 in \citeN{Elstrodt99} and the elementary inequality
$|\text{Im} f|+|\text{Re} f|\le 2|f|$, we get that we can differentiate under
the integral sign if there exists an integrable function $\wp$ such
that for all $u\in\R$ and all $S_0>0$
\begin{align*}
\Big|\frac{\partial}{\partial S_0}F(u,S_0)\Big|\le \wp(u),
\end{align*}
where
\begin{align*}
 F(u,S_0) = S_0^{R-iu} M_{X_T}(R-iu) \widehat{f}(u+iR).
\end{align*}
Now we can estimate the partial derivative of the function $F$:
\begin{align}\label{D-est}
\Big|\frac{\partial}{\partial S_0}F(u,S_0)\Big|
 &= |\e^{(R-1-iu)\log S_0}||R-iu| |M_{X_T}(R-iu) \widehat{f}(u+iR)|\nonumber\\
 &\le \mathpzc{c}(1+|u|)|M_{X_T}(R-iu)||\widehat{f}(u+iR)|
  =:\wp(u).
\end{align}
Analogously we can estimate for the second derivative of $F$:
\begin{align}\label{G-est}
\Big|\frac{\partial^2}{\partial S_0^2}F(u,S_0)\Big|
  &\le \mathpzc{c}'(1+|u|^2)|M_{X_T}(R-iu)||\widehat{f}(u+iR)|
   =:\wp'(u).
\end{align}

Sufficient conditions for the function $\wp$ in \eqref{D-est}, resp.
$\wp'$ in \eqref{G-est}, to be integrable are that
$|u||M_{X_T}(R-iu)|$, resp. $|u|^2|M_{X_T}(R-iu)|$, is integrable and
$\widehat{f}(\cdot+iR)$ is bounded; the first condition dictates in particular
that the measure $P_{X_T}$ -- equivalently $\varrho$ -- has a density of class
$C^1$, resp. $C^2$; see Proposition 28.1 in \citeN{Sato99}. Alternatively, a
sufficient condition is that the function $|u||\widehat{f}(u+iR)|$, resp.
$|u|^2|\widehat{f}(u+iR)|$, is integrable and $M_{X_T}(R-i\cdot)$ is bounded,
highlighting once again the interplay between the properties of the measure
and the payoff function.

\section{Examples of payoff functions}
\label{ch3:payoffs}

\subsection*{1.}
Here we list some representative examples of payoff functions used in
finance, together with their Fourier transforms and comment on whether
they satisfy some of the required assumptions for option pricing. The
calculations for the call option are provided explicitly and for other
options they follow  along the same lines.

\begin{example}[Call and put option]\label{ch3:ex-call}
The payoff of the standard call option with strike $K\in\Rp$ is $f(x)=(\e^x-K)^+$.
Let $z\in\C$ with $\Im z\in(1,\infty)$, then the
Fourier transform of the payoff function of the call option is
\begin{align}\label{ch3:eqop:1}
 \widehat{f}(z)
  &= \int_\R\e^{izx}(\e^x-K)^+\dx
   = \int_{\ln K}^\infty \e^{(1+iz)x}\dx -K\int_{\ln K}^\infty \e^{izx}\dx\nonumber\\
  &= -K^{1+iz}\frac{1}{1+iz} + K^{iz}\frac{K}{iz}
   = \frac{K^{1+iz}}{iz(1+iz)}.
\end{align}
Now, regarding the dampened payoff function of the call option, we
easily get for $R\in(1,\infty)$ that $g\in\labc\cap L^2(\R)$. The weak
derivative of $g$ is
\begin{align}
 \partial g(x) =
\left\{%
\begin{array}{ll}
    0, & \hbox{if $x<\ln K$,} \\
    \e^{-Rx}(\e^x-R\e^x+RK), & \hbox{if $x>\ln K$.} \\
\end{array}%
\right.
\end{align}
Again, we have that $\partial g\in L^2(\R)$. Therefore, $g\in H^1(\R)$ and
using Lemma \ref{sobo} we can conclude that $\widehat{g}\in\la$. Summarizing,
conditions (C1) and (C3) of Theorem \ref{valuation} are fulfilled for the
payoff function of the call option.

Similarly, for a put option, where $f(x)=(K-\e^x)^+$, we
have that
\begin{align}\label{ch3:eqop:2}
 \widehat{f}(z) = \frac{K^{1+iz}}{iz(1+iz)},
  \qquad  \Im z\in (-\infty,0).
\end{align}
Analogously to the case of the call option, we can conclude for the dampened
payoff function of the put option that $g\in\labc$ and $g\in H^1(\R)$ for $R<0$,
yielding $\widehat{g}\in\la$. Hence, conditions (C1) and (C3) of Theorem
\ref{valuation}  are also fulfilled for the payoff function of the put option.
\end{example}

\begin{example}[Digital option]\label{ch3:ex-dig}
The payoff of a digital call option with barrier $B\in\Rp$ is
$1_{\{\e^x>B\}}$. Let $z\in\C$ with $\Im z\in(0,\infty)$, then
the Fourier transform of the payoff function of the digital call
option is
\begin{align}\label{ch3:eqop:3}
 \widehat{f}(z)
   = -\frac{B^{iz}}{iz}.
\end{align}
Similarly, for a digital put option, where $f(x)=1_{\{\e^x<B\}}$,
we have that
\begin{align}\label{ch3:eqop:4}
 \widehat{f}(z) = \frac{B^{iz}}{iz},
  \qquad  \Im z\in (-\infty,0).
\end{align}
For the dampened payoff function of the digital call and put option,
we can easily check that $g\in\la$ for $R\in(0,\infty)$
and $R\in(-\infty,0)$.

Regarding the continuity and bounded variation properties of the value
function, we have that
\begin{align*}
\mathbb{V}_f(X,\logs)
 &= E\big[1_{\{\e^{X_T-\logs}>B\}}\big]
  = P\big(X_T>\log(B)+\logs\big)
  = 1-F_{X_T}\big(\log(B)+\logs\big),
\end{align*}
where $F_{X_T}$ denotes the cumulative distribution function of $X_T$.
Therefore, $\logs\mapsto\mathbb{V}_f(X,\logs)$ is monotonically decreasing,
hence it has locally bounded variation. Moreover, we can conclude
that $\logs\mapsto\mathbb{V}_f(X,\logs)$ is continuous if the
measure $P_{X_T}$ is \textit{atomless}.

Summarizing, condition (D1) is always satisfied for the payoff function of the
digital option, while the prerequisites of Theorem \ref{valuation-dc} on
continuity and bounded variation are satisfied if the measure $P_{X_T}$ does
not have atoms.
\end{example}

\begin{example}
A variant of the digital option is the so-called
\emph{asset-or-nothing digital}, where the option holder receives
one unit of the \emph{asset}, instead of \emph{currency}, depending
on whether the underlying reaches some barrier or not. The payoff of
the asset-or-nothing digital call option with barrier $B\in\Rp$ is
$f(x)=\e^x1_{\{\e^x>B\}}$, and the Fourier transform, for $z\in\C$
with $\Im z\in(1,\infty)$, is
\begin{align}
 \widehat{f}(z)
   = -\frac{B^{1+iz}}{1+iz}.
\end{align}
Arguing analogously to the previous example, we can deduce that condition
(D1) is always satisfied for the payoff function of the asset-or-nothing
digital option, while the prerequisites of Theorem \ref{valuation-dc} are
satisfied if the measure $P_{X_T}$ does not have atoms.
\end{example}

\begin{example}[Double digital option]\label{ch3:ddig}
The payoff of the double digital call option with barriers
$\underline{B},\overline{B}>0$ is
$1_{\{\underline{B}<\e^x<\overline{B}\}}$. Let $z\in\C\setminus\lbrace0\rbrace$,
then the Fourier transform of the payoff function is
\begin{align}
 \widehat{f}(z)
  = \frac{1}{iz}\left(\overline{B}^{iz}-\underline{B}^{iz}\right).
\end{align}
The dampened payoff function of the double digital option
satisfies $g\in\la$ for all $R\in\R$.

Moreover, we can decompose the value function of the double
digital option as
\begin{align*}
\mathbb{V}_f(X,\logs)
 &= \mathbb{V}_{f_1}(X,\logs)-\mathbb{V}_{f_2}(X,\logs),
\end{align*}
where $f_1(x)=1_{\{\e^{x}<\overline{B}\}}$ and
$f_2(x)=1_{\{\underline{B}\le\e^{x}\}}$. Hence, by the results of Example
\ref{ch3:ex-dig}, we get that condition (D1) is always satisfied for the
payoff function of the double digital option, while the prerequisites of
Theorem \ref{valuation-dc} are satisfied if the measure $P_{X_T}$ does not
have atoms.
\end{example}

\begin{example}[Self-quanto and power options]
The payoff of a self-quanto call option with strike $K\in\Rp$ is
$f(x)=\e^x(\e^x-K)^+$. The Fourier transform of the payoff function of the
self-quanto call option, for $z\in\C$ with $\Im z\in(2,\infty)$, is
\begin{align}
 \widehat{f}(z)
  &= \frac{K^{2+iz}}{(1+iz)(2+iz)}.
\end{align}
The payoff of a power call option with strike $K\in\Rp$ and power $2$ is
$f(x)=[(\e^x-K)^+]^2$; for $z\in\C$ with $\Im z\in(2,\infty)$, the Fourier
transform is
\begin{align}
\widehat{f}(z)
  &= -\frac{2K^{2+iz}}{iz(1+iz)(2+iz)}.
\end{align}
The payoff functions for the respective put options are defined in the
obvious way, while the Fourier transforms are identical, with the range
for the imaginary part of $z$ being respectively $(-\infty,1)$ and
$(-\infty,0)$.

Analogously to Example \ref{ch3:ex-call}, we can deduce that conditions
(C1) and (C3) of Theorem \ref{valuation} are fulfilled for the payoff
function of the self-quanto and the power option.
\end{example}

\begin{remark}
For power options of higher order we refer to Raible \citeyear[Chapter 3]{Raible00}.
\end{remark}

\subsection*{2.}
Next we present some examples of payoff functions for options on
several assets and for options on multiple functionals of one asset,
together with their corresponding Fourier transforms.

\begin{example}[Option on the minimum/maximum]\label{min-two-asset}
The payoff function of a call option on the minimum of $d$ assets is
\begin{align*}
 f(x)=(\e^{x_1}\wedge\dots\wedge\e^{x_d}-K)^+,
\end{align*}
for $x\in\R^d$. The Fourier transform of this payoff function is
\begin{align}\label{Fou-min-d}
\widehat{f}(z)
   = -\frac{K^{1+i\sum_{k=1}^d z_k}}{(-1)^d(1+i\sum_{k=1}^d z_k)\prod_{k=1}^d (iz_k)},
\end{align}
where $z\in\C^d$ with $\Im z_k>0$ for $1\le k\le d$ and
$\Im(\sum_{k=1}^d z_k)>1$; for more details we refer to Appendix
\ref{fou-min}. Then, we can easily deduce for the dampened payoff function
that $g\in L^1_{\text{bc}}(\R^d)$.

Moreover, for the put option on the maximum of $d$ assets, the payoff
function is
\begin{align*}
 f(x)=(K-\e^{x_1}\vee\dots\vee\e^{x_d})^+,
\end{align*}
for $x\in\R^d$, where $a\vee b=\max\{a,b\}$. The Fourier transform is
\begin{align}\label{Fou-max-d}
\widehat{f}(z)
 = \frac{K^{1+i\sum_{k=1}^d z_k}}{(1+i\sum_{k=1}^d z_k)\prod_{k=1}^d (iz_k)},
\end{align}
with the restriction now being $\Im z_k<0$ for all $1\le k\le d$. Again, we can
easily deduce that the dampened payoff function satisfies
$g\in L^1_{\text{bc}}(\R^d)$. Therefore, condition (A1) of Theorem
\ref{valuation-Rd} is satisfied.
\end{example}

\begin{example}
A natural example of multi-asset payoff functions are products of single
asset payoff functions. These payoff functions have the form
\begin{align*}
 f(x) = \prod_{\icc=1}^d f_\icc(x_\icc),
\end{align*}
for $x\in\R^d$, where $x_\icc\in\R$ and $f_\icc:\R\rightarrow\R_+$, for all
$1\le\icc\le d$; for example, one can consider $f_1(x_1)=(\e^{x_1}-K)^+$
and $f_2(x_2)=1_{\{\e^{x_2}>B\}}$.

The Fourier transform of these payoff functions is simply the product
of the Fourier transform of the `marginal' payoff functions, since
\begin{align*}
\widehat{f}(z)
 &= \int_{\R^d} \e^{i\langle z,x\rangle}\prod_{\icc=1}^d f_\icc(x_\icc)\dx
  = \prod_{\icc=1}^d\int_{\R} \e^{iz_\icc x_\icc}f_\icc(x_\icc)\ud x_\icc
  = \prod_{\icc=1}^d\widehat{f}_\icc(z_\icc),
\end{align*}
for $z\in\C^d$ and $z_\icc\in\C$, with $\Im z$ in an appropriate range such
that $\widehat{g}\in L^1(\R^d)$. This range, as well as other properties of
$\widehat{f}$, are dictated by the corresponding properties of the Fourier
transforms $\widehat{f}_\icc$ of the marginal payoff functions $f_\icc$.
\end{example}

\begin{remark}
Further examples of multiple asset payoff functions, such as basket
and spread options, and their Fourier transforms can be found in
\citeN{HubalekKallsen03}.
\end{remark}

\subsection*{3.}
We add a short remark on the rate of decay of the Fourier transform of
the various payoff functions and its consequence for numerical implementations.

Consider the standard call option, where the Fourier transform of
the dampened payoff function has the form, cf. \eqref{ch3:eqop:1},
\begin{align*}
\widehat{g}(u) = \frac{K^{1-R}\e^{iu\log{K}}}{(R-iu)(R-1-iu)},
  \qquad u\in\R.
\end{align*}
Then, we have that
\begin{align*}
|\widehat{g}(u)|
 \leq \frac{K^{1-R}}{\sqrt{R^2+u^2}\sqrt{(R-1)^2+u^2}}
 \leq \frac{K^{1-R}}{(R-1)^2+u^2},
\end{align*}
which shows that $\widehat{g}(u)$ behaves like $\frac{1}{u^2}$ for
$|u|>1$. On the other hand, a similar calculation for the digital
option shows that the Fourier transform of the dampened digital
payoff behaves like $\frac{1}{u}$ for $|u|>1$.

Therefore, splitting a call option into the difference of an
asset-or-nothing digital and a digital option, as many authors have
proposed in the literature (cf. e.g. \citeNP{Heston93}), is not only
`conceptually' sub-optimal, as can be seen by Theorems \ref{valuation}
and \ref{valuation-dc}. More importantly, it is also not optimal from
the numerical perspective, since the rate of decay for the digital
option is much slower than for the call option, leading to slower
numerical evaluation of the corresponding option prices.

Indeed, we have calculated the prices of call options corresponding
to 11 strikes and 10 maturities, first using the formula for the call
option, and then representing the call option as the difference of two
digital options. The numerical calculation using the second method
lasts twice as long (6 secs compared to less than 3 secs) in a
standard Matlab implementation.

\section{Examples of driving processes}
\label{LA}

The application of Fourier transform valuation formulas in practice requires
the explicit knowledge of the moment generating function of the underlying
random variable. As such, Fourier methods are tailor-made for pricing European
options in \lev and affine models, since in these models one typically knows
the moment generating function explicitly (at least up to the solution of a
Riccati equation). In order to give a flavor, we present here an overview of
\lev and affine processes, referring to the literature for specific formulas
and proofs.

In \lev processes, the moment generating function of the random variable is
described by the celebrated \lk formula; for a \lev process \prozess[H] with
triplet ($b,c,\lambda$) we have:
\begin{align}
E\big[\e^{\scal{u}{H_t}}\big] = \exp\left(\kappa(u)\cdot t\right),
\end{align}
for suitable $u\in\R^d$, where the cumulant generating function is
\begin{align}\label{lk-mgf-general}
\kappa(u) = \scal{b}{u} + \frac{1}{2}\scal{u}{cu}
          + \int_{\R^d}\left(\e^{\scal{u}{x}}-1-\scal{u}{h(x)}\right)\lambda(\dx);
\end{align}
here $h$ denotes a suitable truncation function. The most popular \lev models
are the VG and CGMY processes (cf. \citeNP{MadanSeneta90}, 
Carr et al. \citeyearNP{Carretal02}), the hyperbolic, NIG and GH processes
(cf. Eberlein and Keller \citeyearNP{EberleinKeller95},
\citeNP{Barndorff-Nielsen98}, \citeNP{Eberlein01a}), and the Meixner model (cf.
\citeNP{SchoutensTeugels98}).

In affine processes, the moment generating functions are described by the
very definition of these processes. Let \prozess[X] be an affine process
on the state space $D=\R^m\times\Rp^n\subseteq\R^d$, starting from $x\in D$;
i.e., under suitable conditions, there exist functions
$\phi:[0,T]\times\mathcal{I}\to\R$ and
$\psi:[0,T]\times\mathcal{I}\to\R^d$ such that
\begin{align}
E_x\big[\e^{\scal{u}{X_t}}\big] = \exp\left(\phi_t(u) + \scal{\psi_t(u)}{x} \right),
\end{align}
for all $(t,u,x)\in[0,T]\times\mathcal{I}\times D$, $\mathcal{I}\subseteq\R^d$.
The functions $\phi$ and $\psi$ satisfy generalized Riccati equations, while
their time derivatives
\begin{align*}
 F(u) = \frac{\partial}{\partial t}\big|_{t=0}\phi_t(u)
\quad\text{ and }\quad
 R(u) = \frac{\partial}{\partial t}\big|_{t=0}\psi_t(u),
\end{align*}
are of \lk form \eqref{lk-mgf-general}; we refer to
\shortciteN{DuffieFilipoviSchachermayer03} and 
Keller-Ressel \citeyear{KellerRessel08} for
comprehensive expositions and the necessary details. The class of affine
processes contains as special cases -- among others -- many stochastic
volatility models, such as the \citeN{Heston93} model, the BNS model (cf.
Barndorff-Nielsen and Shephard
\citeyearNP{Barndorff-NielsenShephard01}, \citeNP{NicolatoVenardos03}), and
time-changed \lev models (cf. \shortciteNP{Carretal03}, \citeNP{Kallsen06}).

\section{Numerical illustration}
\label{ch3:sc5}

As an illustration of the applicability of Fourier-based valuation formulas
even for the valuation of options on several assets, we present a numerical
example on the pricing of an option on the minimum of two assets. As driving
motions we consider a 2d normal inverse Gaussian (NIG) \lev process and a
2d affine stochastic volatility model.

Let $H$ denote a 2d NIG random variable,
i.e. \[H=(H^1,H^2)\sim\text{NIG}_2(\alpha,\beta,\delta,\mu,\Delta),\]
where the parameters have the following domain of definition:
$\alpha,\delta\in\R_+$, $\beta,\mu\in\R^2$, and
$\Delta\in\R^{2\times2}$ is a symmetric, positive-definite matrix; w.l.o.g.
we can assume that $\text{det}(\Delta)=1$; in addition,
$\alpha^2>\scal{\beta}{\Delta\beta}$. Then, the moment generating function of
$H$, for $u\in\R^2$ with $\alpha^2-\langle \beta+u,\Delta(\beta+u)\rangle\ge0$,
is
\begin{align}\label{2d-NIG-MGF}
M_{H}(u)
 &= \exp\left( \langle u,\mu\rangle
       + \delta\left(\sqrt{\alpha^2-\langle \beta,\Delta\beta\rangle}
           -\sqrt{\alpha^2-\langle \beta+u,\Delta(\beta+u)\rangle}\right)\right).
\end{align}
In the $\text{NIG}_2$ model, we specify the parameters $\alpha,\beta,\delta$
and $\Delta$, and the drift vector $\mu$ is determined by the martingale
condition. Note that the marginals $H^\icc$ are also NIG
distributed (cf. \citeNP[Theorem 1]{Blaesild81}), hence the drift vector can be
easily evaluated from the cumulant of the univariate NIG law. The covariance
matrix corresponding to the $\text{NIG}_2$-distributed random variable $H$ is
\begin{align*}
\Sigma_{\text{NIG}}
 &= \delta \left( \alpha^2-\scal{\beta}{\Delta\beta} \right)^{-\frac{1}{2}}
    \left( \Delta + \left( \alpha^2-\scal{\beta}{\Delta\beta} \right)^{-1}
           \Delta\beta\beta^{\top}\Delta\right),
\end{align*}
cf. \citeN[eq. (4.15)]{Prause99}. A comprehensive exposition of the
multivariate generalized hyperbolic distributions can be found in
\citeN{Blaesild81}; cf. also \citeN{Prause99}.

We will also consider the following affine stochastic volatility model
introduced by \citeN{DempsterHong02}, that extends the framework of Heston
\citeyear{Heston93} to the multi-asset case. Let $H=(H^1,H^2)$ denote the logarithm
of the asset price processes $S=(S^1,S^2)$, i.e. $H^\icc=\log S^\icc$; then,
$H^\icc$, $\icc=1,2$ satisfy the following SDEs:
\begin{align*}
\ud H^1_t &= -\frac{1}{2}\sigma_1^2v_t\dt + \sigma_1\sqrt{v_t}\ud W_t^1 \nonumber\\
\ud H^2_t &= -\frac{1}{2}\sigma_2^2v_t\dt + \sigma_2\sqrt{v_t}\ud W_t^2\\
\ud v_t &= \kappa(\mu-v_t)\dt + \sigma_3\sqrt{v_t}\ud W_t^3, \nonumber
\end{align*}
with initial values $H_0^1,H_0^2,v_0>0$. The parameters have the
following domain of definition: $\sigma_1,\sigma_2,\sigma_3>0$ and
$\mu,\kappa>0$. Here $W=(W^1,W^2,W^3)$ denotes a 3-dimensional
Brownian motion with correlation coefficients
\begin{align*}
\scal{W^1}{W^2}=\rho_{12},\quad
\scal{W^1}{W^3}=\rho_{13},\quad\text{ and }\quad
\scal{W^2}{W^3}=\rho_{23}.
\end{align*}
The moment generating function of the vector $H=(H^1,H^2)$ has been calculated
by \citeN{DempsterHong02}; for $u=(u_1,u_2)\in\R^2$ we have
\begin{align*}
M_{H_t}(u)
 &= \exp\bigg( \scal{u}{H_0}
              + \frac{2\zeta(1-\e^{-\theta t})}{2\theta - (\theta-\gamma)(1-\e^{-\theta t})}\cdot v_0
  \nonumber\\
  &\qquad\qquad - \frac{\kappa\mu}{\sigma_3^2}
            \Big[ 2\cdot\log\Big(\frac{2\theta -
                 (\theta-\gamma)(1-\e^{-\theta t})}{2\theta}\Big) + (\theta-\gamma)t \Big] \bigg),
\end{align*}
where $\zeta=\zeta(u)$, $\gamma=\gamma(u)$, and $\theta=\theta(u)$ are
\begin{align*}
\zeta &=
 \frac{1}{2}\Big(\sigma_1^2u_1^2 + \sigma_2^2u_2^2 + 2\rho_{12}\sigma_1\sigma_2u_1u_2
                  - \sigma_1^2u_1 - \sigma_2^2u_2\Big),\\
\gamma &=
 \kappa - \rho_{13}\sigma_1\sigma_3u_1 - \rho_{23}\sigma_2\sigma_3u_2,\\
\theta &=
 \sqrt{\gamma^2 - 2\sigma_3^2\zeta}.
\end{align*}

We can deduce that all three models satisfy conditions (A2) and (A3) of Theorem
\ref{valuation-Rd} for certain values of $R$. Explicit calculations for the 2d
NIG model are deferred to Appendix \ref{regularity-2D-NIG}; analogous calculations
yield the results for the other models.

The Fourier transform of the payoff function
$f(x)=(\e^{x_1}\wedge\e^{x_2}-K)^+$, $x\in\R^2$, corresponding to the option on
the minimum of two assets is given by \eqref{Fou-min-d} for $d=2$, and we get
that condition (A1) of Theorem \ref{valuation-Rd} is satisfied for $R_1,R_2>0$
such that $R_1+R_2>1$.

Therefore, applying Theorem \ref{valuation-Rd}, the price of an option on the
minimum of two assets is given by
\begin{align*}
\mathbb{MTA}_T(S^1,S^2;K)
 &= \frac{1}{4\pi^2}\int_{\R^2}
    (S_0^1)^{R_1+iu_1}(S_0^2)^{R_2+iu_2}
    M_{H_T}(R_1+iu_1,R_2+iu_2)\nonumber\\
 & \qquad \times
   \frac{K^{1-R_1-R_2-iu_1-iu_2}}{(R_1+iu_1)(R_2+iu_2)(R_1+R_2-1+iu_1+iu_2)}\ud u,
\end{align*}
where $M_{H_T}$ denotes the moment generating function of the random vector
$H_T$, and $R_1,R_2$ are suitably chosen.

In the numerical illustrations, we consider the following parameters: strikes
\[
K=\left\{85,90,92.5,95,97.5,100,102.5,105,107.5,110,115\right\}
\]
and times to maturity
\[
T = \left\{\tfrac{1}{12},\tfrac{2}{12},0.25,0.50,0.75,1.00\right\}.
\]
In the 2d NIG model, we consider some typical parameters, e.g.
$S_0^1=100$, $S_0^2=95$, $\alpha=6.20$, $\beta_1=-3.80$, $\beta_2=-2.50$ and
$\delta = 0.150$; we consider two matrices
$\Delta^+=\bigl(\begin{smallmatrix}1&0\\0&1\end{smallmatrix}\bigr)$ and
$\Delta^-=\bigl(\begin{smallmatrix}1&-1\\-1&2\end{smallmatrix}\bigr)$, which
give positive and negative correlations respectively; indeed we get that
\begin{align*}
\Sigma_{\text{NIG}}^+ =
\begin{pmatrix}
    0.0646  &  0.0191\\
    0.0191  &  0.0481
\end{pmatrix}
\quad \text{and} \quad
\Sigma_{\text{NIG}}^- =
\begin{pmatrix}
    0.0287  & -0.0258\\
   -0.0258  &  0.0556
\end{pmatrix}.
\end{align*}
The option prices in these two cases are exhibited in Figure \ref{Fig:2NIG}.

\begin{figure}
\begin{center}
 \includegraphics[width=6.250cm,keepaspectratio=true]{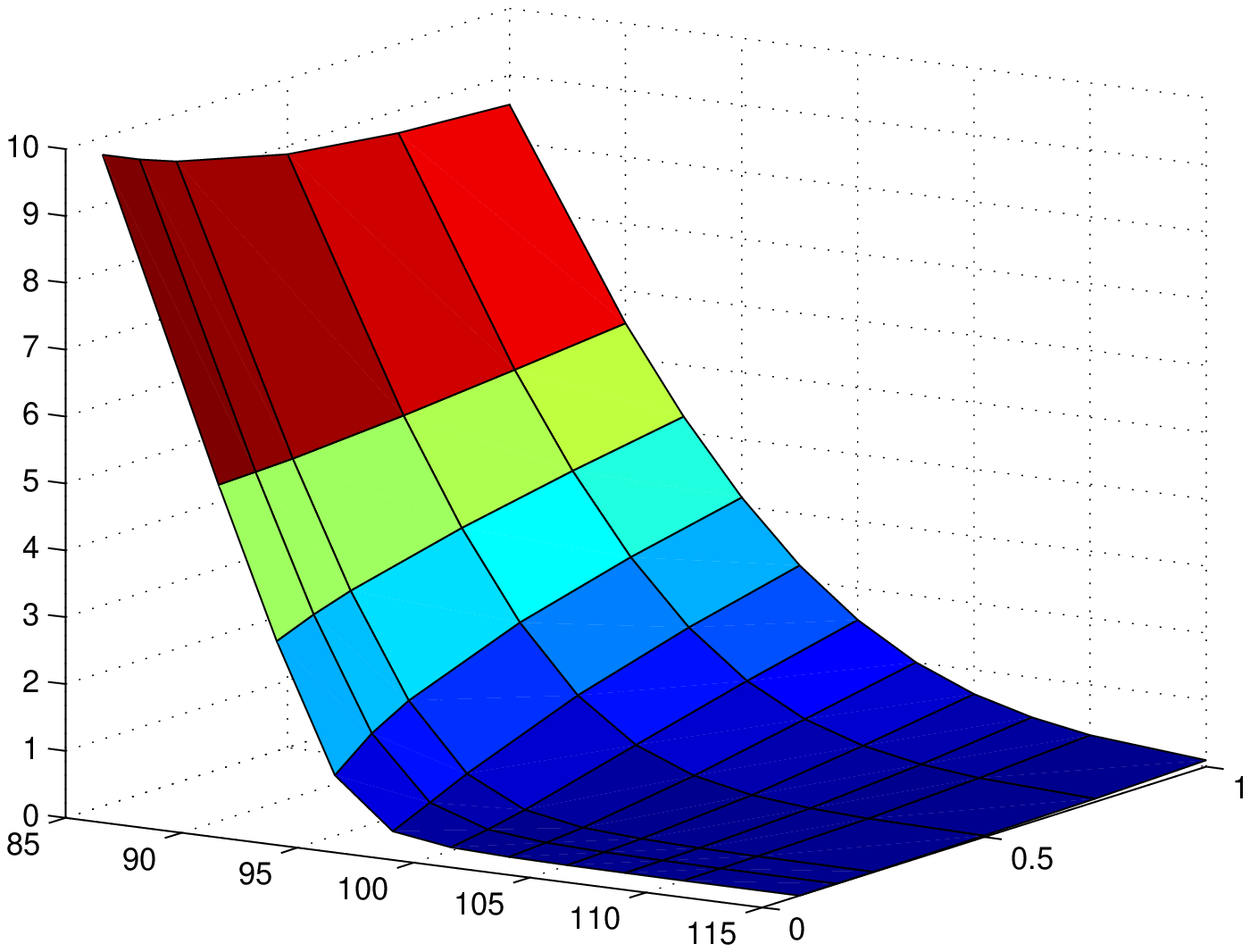}
 \includegraphics[width=6.250cm,keepaspectratio=true]{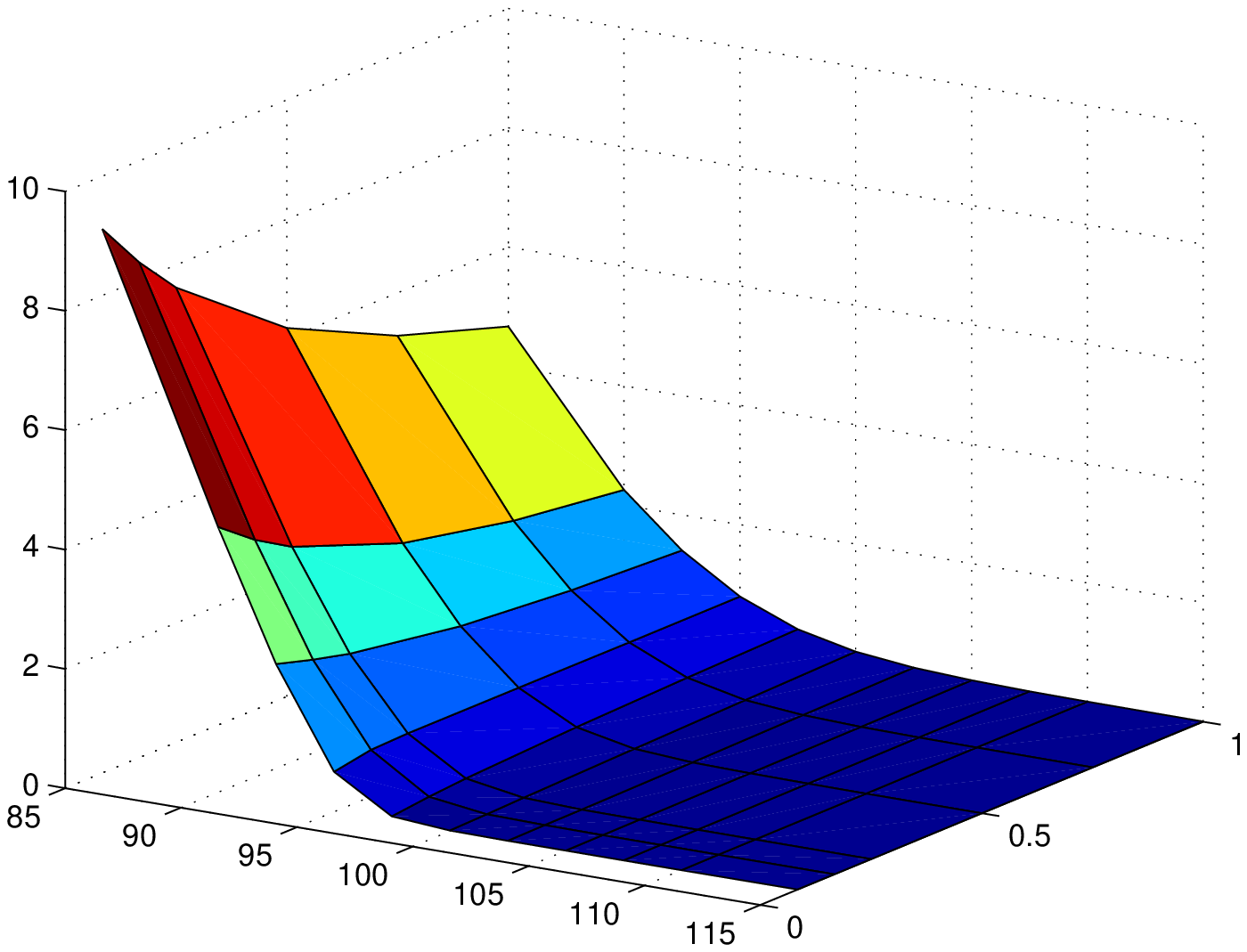}
  \caption{Option prices in the 2d NIG model with positive (left) and
           negative (right) correlation.}
  \label{Fig:2NIG}
\end{center}
\end{figure}

Finally, in the stochastic volatility model we consider the parameters used in
\citeN{DempsterHong02}, that is $S_0^1=96$, $S_0^2=100$, $\sigma_1=0.5$,
$\sigma_2=1.0$, $\sigma_3=0.05$, $\rho_{12}=0.5$, $\rho_{13}=0.25$,
$\rho_{23}=-0.5$, $v_0=0.04$, $\kappa=1.0$ and $\mu=0.04$; the option prices are
shown in Figure \ref{Fig:2SV}.

\begin{figure}
\begin{center}
 \includegraphics[width=8.00cm,keepaspectratio=true]{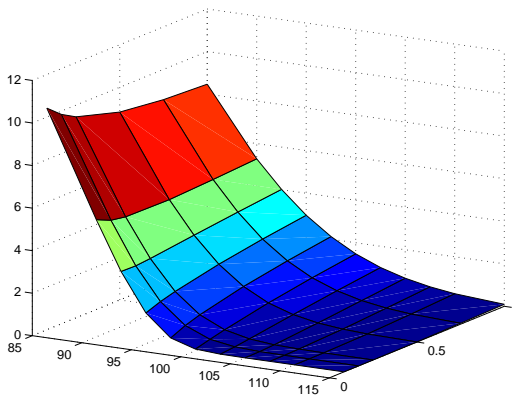}
  \caption{Option prices in the 2d stochastic volatility model.}
  \label{Fig:2SV}
\end{center}
\end{figure}

\appendix
\section{Fourier transforms of multi-asset options}
\label{fou-min}

In this appendix we outline the derivation of the Fourier transform
corresponding to the payoff function of an option on the minimum of
several assets; the derivation for the maximum is completely analogous
and therefore omitted.

The payoff of a (call) option on the minimum of $d$ assets is
\begin{align*}
(S^1\wedge S^2\wedge\dots\wedge S^d -K)^+.
\end{align*}
The payoff function $f$ corresponding to this option is given,
for $x\in\R^d$, by
\begin{align*}
f(x)
 &=(\e^{x_1}\wedge\e^{x_2}\wedge\dots\wedge\e^{x_d} -K)^+
  = (\e^{x_1\wedge x_2\wedge\dots\wedge x_d}-K)^+. 
\end{align*}
The following decomposition holds, if $x_i\neq x_j$ for $i\neq j$,
$1\le i,j\le d$
\begin{align*}
f(x)
  = \sum_{i=1}^d(\e^{x_i}-K)^+1_{\{x_i\le x_j, \forall j\}}
  = \sum_{i=1}^d(\e^{x_i}-K)
    \prod_{\substack{j=1\\j\neq i}}^d1_{\{k<x_i<x_j\}},
\end{align*}
where $k=\log K$. Define also the auxiliary functions $f_i$, $1\le i\le d$, where
\begin{align*}
f_i(x)= (\e^{x_i}-K)\prod_{\substack{j=1\\j\neq i}}^d1_{\{k<x_i<x_j\}}.
\end{align*}

The dampened payoff function is $g(x)=\e^{-\langle R,x\rangle}f(x)$,
where $R\in\R^d$; we define analogously the dampened $f_i$-functions,
i.e. $g_i(x)=\e^{-\langle R,x\rangle}f_i(x)$. For simplicity, we first
calculate the Fourier transform of the dampened $f_1$-function; for
$u\in\R^d$ we get
\begin{align}
\widehat{g}_1(u)
 &= \int_{\R^d}\e^{\langle iu-R,x\rangle}(\e^{x_1}-K)\prod_{j=2}^d1_{\{k<x_1\le x_j\}}\dx\nonumber\\
 &= \int_k^{\infty}\int_{x_1}^{\infty}\dots\int_{x_1}^{\infty}
    \e^{\langle iu-R,x\rangle}(\e^{x_1}-K)\dx_d\dots\dx_1 \nonumber\\
 &= \int_k^{\infty}\e^{(iu_1-R_1)x_1}(\e^{x_1}-K)
    \left(\prod_{j=2}^d\int_{x_1}^{\infty}\e^{(iu_j-R_j)x_j}\dx_j\right)\dx_1 \nonumber\\
 &= \int_k^{\infty}\e^{(iu_1-R_1)x_1}(\e^{x_1}-K)
    \prod_{j=2}^d\left(-\frac{\e^{(iu_j-R_j)x_1}}{iu_j-R_j}\right)\dx_1\nonumber\\
 &= \frac{1}{\prod_{j=2}^d(R_j-iu_j)}
    \int_k^{\infty}\e^{\sum_{j=1}^d(iu_j-R_j)x_1}(\e^{x_1}-K)\dx_1 \nonumber\\
 &= \frac{1}{\prod_{j=2}^d(R_j-iu_j)}\left(
     -\frac{K^{1+\sum_{j=1}^d(iu_j-R_j)}}{1+\sum_{j=1}^d(iu_j-R_j)}
     +\frac{K^{1+\sum_{j=1}^d(iu_j-R_j)}}{\sum_{j=1}^d(iu_j-R_j)}\right)\nonumber\\
 &= \frac{K^{1+\sum_{j=1}^d(iu_j-R_j)}}
         {\prod_{j=2}^d(R_j-iu_j)\times\left(1+\sum_{j=1}^d(iu_j-R_j)\right)
           \times\left(\sum_{j=1}^d(iu_j-R_j)\right)} \nonumber,
\end{align}
subject to the conditions $R_j>0$ for all $j\ge2$ and
$\sum_{j=1}^dR_j>1$.

Hence, in general we have that
\begin{align*}
\widehat{g}_l(u)
 &= \frac{K^{1+\sum_{j=1}^d(iu_j-R_j)}}
         {\prod_{\substack{j=1\\j\neq l}}^d(R_j-iu_j)
          \times\left(1+\sum_{j=1}^d(iu_j-R_j)\right)\times\left(\sum_{j=1}^d(iu_j-R_j)\right)},
\end{align*}
subject to the conditions $R_j>0$ for all $1\le j\le d$ and
$\sum_{j=1}^dR_j>1$.

Now, we recall that $f(x)=\sum_{l=1}^df_l(x)$, hence $g(x)=\sum_{l=1}^dg_l(x)$
which yields $\widehat{g}(u)=\sum_{l=1}^d\widehat{g}_l(u)$; therefore
\begin{align*}
\widehat{g}(u)
 &= \sum_{l=1}^d\frac{K^{1+\sum_{j=1}^d(iu_j-R_j)}}
         {\prod_{\substack{j=1\\j\neq l}}^d(R_j-iu_j)
          \times\left(1+\sum_{j=1}^d(iu_j-R_j)\right)\times\left(\sum_{j=1}^d(iu_j-R_j)\right)} \nonumber\\
 &= \frac{K^{1+\sum_{j=1}^d(iu_j-R_j)}}
         {\left(1+\sum_{j=1}^d(iu_j-R_j)\right)\times\left(\sum_{j=1}^d(iu_j-R_j)\right)}
    \;\sum_{l=1}^d\frac{R_l-iu_l}{\prod_{j=1}^d(R_j-iu_j)}\nonumber\\
 &= \frac{-K^{1+\sum_{j=1}^d(iu_j-R_j)}}
         {(-1)^d\prod_{j=1}^d(iu_j-R_j)\left(1+\sum_{j=1}^d(iu_j-R_j)\right)}.
\end{align*}
This we can also rewrite as
\begin{align}
\widehat{f}(z)
 &= - \frac{K^{1+i\sum_{j=1}^d z_j}}
           {(-1)^d\prod_{j=1}^d(iz_j)\left(1+i\sum_{j=1}^dz_j\right)},
\end{align}
subject to the conditions $\Im z_j>0$ for all $1\le j\le d$ and
$\sum_{j=1}^d \Im z_j>1$.

\section{Calculations for the 2d NIG model}
\label{regularity-2D-NIG}

By the moment generating function of the 2d NIG process, cf. \eqref{2d-NIG-MGF},
it is evident that assumption (A2) is satisfied for $R\in\R^2$ with
$\alpha^2-\langle\beta+R,\Delta(\beta+R)\rangle\ge0$. In order to verify
condition (A3) we have to show that the function $u\mapsto M_H(R+iu)$ is
integrable; it suffices to show that the real part of the exponent of
$M_H(R+iu)$ decays like $-|u|$. We have
\begin{align*}
\log\big(M_H(R+iu)\big)
 &= i \langle  \mu,u \rangle +  \langle  \mu,R\rangle  + \delta \sqrt{ \alpha^2 - \langle \beta,\Delta \beta \rangle }\\
 &\quad - \delta \sqrt{ \alpha^2 - \langle \beta+R+iu, \Delta(\beta + R +iu)\rangle}\,.
\end{align*}
Recall that the product $\langle\cdot,\cdot\rangle$ over $\C^d$ is defined
as follows: for $u,v\in\C^d$ set $\langle u,v\rangle = \sum_iu_iv_i$. Then
\begin{multline*}
\langle \beta+R+iu, \Delta(\beta + R +iu)\rangle \\ =
 \langle \beta+R, \Delta(\beta + R )\rangle - \langle u, \Delta u\rangle + 2 i \langle \beta+R, \Delta u\rangle
\end{multline*}
and since
$\sqrt{z} = \sqrt{\frac{1}{2}(|z| + \Re(z))} + i \frac{\Im(z)}{|\Im(z)|}\sqrt{\frac{1}{2}(|z| - \Re(z))}$,
we get
\begin{align*}
\lefteqn{ \Re\big(  \log\big( M_H(R+iu)\big)\big)}\\
&=
   \langle  \mu,R\rangle  + \delta \sqrt{ \alpha^2 - \langle
\beta,\Delta \beta \rangle }
-\frac{\delta}{\sqrt{2}}\Big\{
\big| \alpha^2 - \langle \beta+R+iu, \Delta(\beta + R +iu)\rangle \big|  \\
&\quad+ \alpha^2 - \langle \beta+R,\Delta( \beta+R) \rangle + \langle u,
\Delta u\rangle  \Big\}^{1/2} \\
&\le
  \langle  \mu,R\rangle  + \delta \sqrt{ \alpha^2 - \langle \beta,\Delta
\beta \rangle }
-  \delta \sqrt{ \alpha^2 - \langle \beta+R,\Delta (\beta +R)\rangle +
\langle u, \Delta u\rangle  } \\
&\le
\langle  \mu,R\rangle  + \delta \sqrt{ \alpha^2 - \langle \beta,\Delta
\beta \rangle }
- \delta \sqrt{\lambda_{\min}} |u| \,,
\end{align*}
where $\lambda_{\min}$ denotes the smallest eigenvalue of the matrix $\Delta$.

\bibliographystyle{chicago}
\bibliography{references}

\end{document}